\newcommand{\ank}[1]{\mkern 2mu\bm{\Upsilon}\mkern -3mu_{#1}\mkern 2mu}

\newcommand{\open}{\Bbb}

\newcommand{\oN}{{\open N}}

\def\={=\!\!}
\def\an{\mkern 2mu\bm{\Upsilon}\mkern 2mu}
\newcommand*\dep{{=\mkern-1.2mu}}

\documentclass[12pt]{article}
\usepackage{amsthm} 
\usepackage{hyperref}
\usepackage{graphicx} 
\usepackage{amssymb,latexsym,times,bm}
\newcommand{\psfrag}[2]{}
\usepackage{makeidx,proof}
\usepackage{doi}
\usepackage{enumerate}
\newtheorem{theorem}{Theorem}
\newtheorem{definition}[theorem]{Definition}

\newtheorem{example}[theorem]{Example}

\begin{document}

\title{An atom's worth of anonymity\thanks{The  author would like to thank  the Academy of Finland, grant no: 322795, and funding from the European Research Council (ERC) under the
European Union’s Horizon 2020 research and innovation programme (grant agreement No
101020762).}}
\author{Jouko V\"a\"an\"anen\\
\\
University of Helsinki, Finland}
\date{}
\maketitle

I contribute this paper on anonymity to honor the birthday of John Crossley. I am not only John's friend but also his grandson in the academic sense, as my doctoral supervisor was Peter Aczel, a student of John's. Of course, it is my hope that this attempt to find the logic of anonymity amuses John.

Suppose we have data about some individuals and in the data the values of some attributes are publicly known while the values of some other attributes are sensitive personal data that should be protected from becoming public.  {\em Anonymity} in this context  means the impossibility to derive the values of the sensitive attributes from the values of the publicly known attributes, and this impossibility should be true of each individual.

Anonymity has gained notoriety in modern times as data about our actions and choices accumulates in the internet partly unbeknownst to us and partly by our own choice. Usually people wish some data about themselves were private while some other date may be public or is even wanted to be public for publicity reasons. There are different criteria which characterize the degree of anonymity of data. Given data can also be anonymized by different techniques in order to increase its degree of anonymity.

In this paper we take a very simple ``atomic" degree of anonymity as our starting place. We axiomatize these atoms
and propose the investigation of first order logic based on these atoms. Considering the vast literature and the huge importance of anonymity our investigation may seem quite modest. However, it is about the \emph{logic} of anonymity, not about how to secure, create or break anonymity. 

In this paper, as in database theory, variables denote attributes. Thus our $x$, $y$, $z$, etc stand for such attributes as ``date", ``time", ``surname", ``first name", ``hometown", ``year of birth", ``gender", ``salary", ``blood type", ``systolic blood pressure", ``diastolic blood pressure", etc. It is natural to think of an assignment of values to such variables as a database record. We call any set of such assignments a \emph{team.}  Calling them teams carries no special meaning but is just a practice I adopted in my book \cite{MR2351449}, with teams of robots performing some task  in mind.

\begin{table}
\begin{center}\begin{tabular}{llc}
{\tt surname}&{\tt hometown}& {\tt salary}\\
\hline
Balbuk&Watarru&70,000\\
Barambah&Amata&90,000\\
Jones&Finke&100,000\\
Smith&Watarru&70,000\\
Williams&Amata&90,000\\
Yunipingu&Finke&100,000\\
\end{tabular}
\end{center}\caption{A team.\label{p1}}\end{table}

In the team of Table~\ref{p1} the attributes {\tt hometown} and {\tt salary} do not reveal 100\% the attribute 
{\tt surname} because for all combinations of {\tt hometown} and {\tt salary} there are two different values of 
{\tt surname}. So, if we know   {\tt hometown} and {\tt salary} we are still  a little uncertain  who is in question. Of course, this is an extremely weak form of anonymity, and that is why we call it ``atomic" anonymity. On the other hand, the attribute {\tt surname} completely determines {\tt hometown} and {\tt salary}, so no anonymity or uncertainty there.

We formalize the concept of keeping part of data anonymous while publishing some other part. 
We use $x,y,z,u,v,$ etc to denote finite sequences of attributes. The concatenation of $x$ and $y$ is denoted $xy$. If $s$ is an assignment and $x=\langle x_1,\ldots, x_n\rangle$, then   $\langle s(x_1),\ldots,s(x_n)\rangle$ is denoted $s(x)$.

\begin{definition}[Anonymity atom]\label{anatom}
A team $X$ satisfies ``$x$ leaves $y$ anonymous'' in the model $M$, in symbols $M\models_X x\an y$, if $$\forall s\in X\exists s'\in X(s(x)=s'(x)\wedge s(y)\ne s'(y)).$$  
\end{definition}

The intuition of $M\models_X x\an y$ is the following: We can publish $X\restriction x$ and still  keep the anonymity of $X\restriction y$ as there is {\em no} row $s$ of $X$ (i.e. {\em no} individual $s$) for which we can infer $s(y)$ from $s(x)$. A better anonymity is provided by the {\em $k$-anonymity} atom which demands in Definition~\ref{anatom} not just one but $k$ different values $s'(y)\ne s(y)$ for each $s(y)$.

Note that the full team (in a domain with at least two elements) and the empty team always satisfy $x\an y$. Unlike $\dep(x,y)$ (as in \cite{MR2351449}), the atom $x\an y$ is not closed downwards (or upwards). We can regard $x\an y$ as a strong (maximal) denial of $\dep(x,y)$.
But note  that the independence atom $x\ \bot\ y$ (as in \cite{MR3038039}) does {\em not} imply $x\an y$, because $x\ \bot\  y$ may be true because $y$ is constant. The anonymity atom is clearly closed under unions and it is first order. Hence it is definable in inclusion logic \cite{MR3111746}, in fact, if $u$ is a sequence of new variables, as many as in $y$, then

$$x\an y\iff\exists u(\neg u=y\wedge xu\subseteq xy).$$

Thus `anonymity logic' is a sublogic on inclusion logic. In fact, the two are equivalent \cite{galthesis}. See also \cite{MR997296}.

We now ask the question, what are the general principles that propositions of the form $x\an y$ obey? If we know a bunch of such propositions, what can we infer from them? For example,
is it true that if $x\an y$ and $y\an z$, then $x\an z$? That this is not so can be seen by investigating the team
of Table~\ref{p2}. This team satisfies $x\an y$ and $y\an z$ but fails to satisfy  $x\an z$ because it actually satisfies $x=z$.

\begin{table}
$$\begin{array}{ccc}
x&y&z\\
\hline
0&0&0\\
0&1&0\\
1&0&1\\
1&1&1\\
\end{array}$$
\caption{Failure of transitivity\label{p2}}\end{table}
\begin{definition}\label{a1}
The axioms of {\em anonymity} are 
\begin{description}
\item[A1] $xy\an zu$ implies $yx\an zu$ and $xy\an uz$. (Permutation)
\item[A2] $xy\an z$ implies $x\an zu$. (Monotonicity)
\item[A3]  $xy\an zy$ implies $xy\an z$. (Cancellation)
\item [A4]  $x\an $ implies $\bot$.
\end{description}
\end{definition}

From the axioms (or rules) of Definition~\ref{a1} we get a notion $\Sigma\vdash x\an y$ of {\em derivability} of an atom $x\an y$ from a set $\Sigma$ of such atoms.

\begin{theorem}\label{main} The following are equivalent for any set $\Sigma\cup\{x\an y\}$ of anonymity atoms:
\begin{enumerate}
\item $\Sigma\vdash x\an y$. 
\item $M\models_X\Sigma$ implies $M\models_X x\an y$ for every model $M$ and team $X$.
\end{enumerate}
\end{theorem}

\begin{proof} Suppose first $\Sigma\vdash x\an y$ and $M\models_X\Sigma$. All the rules (A1)-(A4) preserve truth. Hence $M\models_X x\an y$.
Suppose then $\Sigma\nvdash x\an y$. By the validity of (A3) we may assume $x\cap y=\emptyset$. By the validity of (A4) we may assume $y\ne\emptyset$. Let $x=\langle x_1,\ldots,x_n\rangle$ and $y=\langle y_1,\ldots,y_m\rangle$. 
Let $W$ be the set of all variables occurring in $\Sigma\cup\{ x\an y\}$. Let $M$ be a model with at least three elements $0,1,2$ and
$$
\begin{array}{lcl}
X&=&\{s\in \{0,1,2\}^W : \bigvee_{1\le i\le n}(s(x_i)\ne 0)\vee
\bigwedge_{1\le i\le m}(s(y_i)= 0)\}
\end{array}$$

First of all,  $M\not\models_X x\an y$ for if $M\models_X x\an y$ and we  let $s(w)=0$ for all $w\in W$, then there would be $s'\in X$ such that $s'(x_i)=s(x_i)=0$ for all $i\in \{1,\ldots,n\}$ and $s'(y_i)\ne s(y_i)=0$  for some $i\in\{1,\ldots,m\}$, a contradiction.
We are left with proving  $M\models_X \Sigma$. To see this, let ${u\an v}\in \Sigma$. We may assume, w.l.o.g., $u\cap v=\emptyset$. Let $s\in X$. We shall construct $s'\in X$ such that $s'(u)=s(u)$ and $s'(v)\ne s(v)$.

Let us assume first:
 ${x\not\subseteq u.}$
Since $x\cap(W\setminus u)\ne \emptyset$, there is some $x_i\notin u$. Either $x_i\in v$ or $x_i\notin v$. \medskip

\noindent{\bf Case 1.1:} $x_i\in v$, i.e. $x_i=v_j$ for some $j\in\{1,\ldots,k\}$. Since there are three possible values for $s'(v_j)$, we can make sure both $s'(v_j)\ne s(v_j)$ and $s'(v_j)\ne 0$. Then $s'(x)\ne 0$, since  $x_i=v_j$. Thus  $s'\in X$, whichever way we otherwise define $s'$. By construction $s'(v)\ne s(v)$. Since $u\cap v=\emptyset$, we are free to let $s'(u)=s(u)$, and we are done. 

\medskip

\noindent{\bf Case 1.2:} $x_i\notin v$. Let $s'(u)=s(u)$ and let $s'(v)$ be such that $s'(v)\ne s(v)$.
Since now $x_i\notin u\cup v$,  we can let $s'(x_i)\ne 0$, and this guarantees $s'\in X$. We are done again.

Let us then assume 
${x\subseteq u.}$
%
%
Since $s\in X$, we know that $s(x)$ is not all zero or $s(y)$ is all zero.
\medskip

\noindent{\bf Case 2.1:} $s(x)$ is not all zero.  We let $s'(u)=s(u)$. Since now $s'(x)$ is not all zero, we have $s'\in X$ whatever $s'$ is otherwise. So we choose $s'(v)\ne s(v)$ in an arbitrary way, and we are done.
So let us assume  
\medskip

\noindent{\bf Case 2.2:} $s(y)$ is all zero. If $v\subseteq y$, then ${u\an v}\in\Sigma$ implies
$x\an y\in \Sigma$, by Axiom (A2), contrary to our assumption. So  we may pick $v_i\in v\setminus y$. Let $s'(v_i)\ne s(v_i)$. Now we know $s'(v)\ne s(v)$. Since $u\cap v=\emptyset$, we can let $s'(u)=s(u)$. Since $v_i\notin y$ we can let $s'(y)=s(y)=0$. Now $s'\in X$, whichever way $s'$ is otherwise defined.
\end{proof}

We may combine $x\an y$ with logical operations as is done in dependence logic \cite{MR2351449}. Here is a practical example:

\begin{example} Suppose we collect some data about individuals and we need to publish this data together with the addresses of the individuals. Some individuals may, for privacy reasons,  not want their address to be associated with other data about them. For this purpose we allow any individual to declare their address private by setting ${\tt publicity}={\tt private}$. A team satisfies, according to the semantics of implication given in    \cite{MR2351449}:
$${\tt publicity}={\tt private}\to {\tt data}\an {\tt address},$$ iff the anonymity constraint ${\tt data}\an {\tt address}$ holds (at least) in the part of the team that  satisfies ${\tt publicity}={\tt private}$. 
\end{example}

\begin{example} We continue the previous example. A team satisfies  
$$\exists x (\phi(x)\wedge ({\tt data}\ x)\an {\tt address}),$$ iff (possibly multiple) values can be assigned to $x$ so that the new values of $x$ satisfy $\phi(x)$ and, in addition, the anonymity constraint $({\tt data}\ x)\an {\tt address}$ (still) holds. In other words, data about the individuals can be updated with data about the attributes $x$ without violating the modest anonymity provided by the $\an$-atom.
\end{example}

Our anonymity atom provides the most basic form of anonymity, not sufficient for practical applications. Therefore we strengthen the atom to $k$-anonymity atom. The bigger the $k$, the more anonymity it provides. The idea is that  an attribute  $y$ is $k$-anonymous with respect to some other attributes $x$ if any data about $x$ in the team is associated with at least $k$ values for $y$. Intuitively, my name does not reveal me if there are $k$ different persons with the same name as me, and $k$ is relatively large. In September 9, 2022, there were 4703 people alive with the last name V\"a\"an\"anen in Finland. Thus I can remain relatively anonymous in Finland even if I reveal that my last name is V\"a\"an\"anen. At the same moment of time there were only 7 people  in Finland's population information system with the last name Crossley, providing much less anonymity if the last name was revealed. 
 
\begin{definition}[$k$-anonymity atom]\label{anatom}
A team $X$ satisfies ``publishing $x$ keeps $y$ $k$-anonymous'' in a model $M$, in symbols 
$M\models_X x\ank{k} y$, if $$\forall s\in X(|\{s'(y) : s'\in X \wedge s(x)=s'(x)\}|\ge k).$$  The atom $x\ank{k} y$ is called {\em simple}, if $|y|=1$.
\end{definition}

Note that the criterion is different from
$$\forall s\in X(|\{s\in X : s(x)=s'(x)\wedge s(y)\ne s'(y)|\ge k)$$ but equivalent to
$$\forall s\in X\exists s'_1\ldots\exists s'_k\in X\bigwedge_{1\le i<j\le k}(s(x)=s'_i(x)\wedge s'_i(y)\ne s'_j(y)).$$  

$x\ank{1}y$ is always true and $x\ank{2} y$ is equivalent to $x\an y$.

\begin{definition}\label{a}
The axioms of {\em $k$-anonymity} (for varying $k$) are 
\begin{description}
\item[A1] $xy\ank{k} zu$ implies $yx\ank{k} zu$ and $xy\ank{k} uz$. (Permutation)
\item[A2] $xy\ank{k} z$ implies $x\ank{l} zu$ for $l\le k$. (Monotonicity)
\item[A3] If $xy\ank{k} zy$, then $xy\ank{k} z$. (Cancellation)
\item [A4] If $x\ank{k} $, then $\bot$.
\item [A5] If $x\ank{k_1} y_1$, $xy_1\ank{k_2}y_2, \ldots,  xy_1...y_{n-1}\ank{k_n} y_n$, then $$x\ank{k_1\cdot\ldots\cdot k_n}y_1y_2\ldots y_n.$$ 
\end{description}
\end{definition}

From the axioms (or rules) of Definition~\ref{a} we get a notion $\Sigma\vdash x\ank{k} y$ of {\em derivability} of an atom $x\an y$ from a set $\Sigma$ of such (for various $k$).

We conjecture that the above axioms are complete in the sense that $\Sigma\vdash x\ank{k} y$ holds iff every team satisfying $\Sigma$ also satisfies $x\ank{k} y$. We now prove the conjecture for the special case of simple $k$-anonymity atoms. Then, of course (A5) is  unnecessary.
 
\begin{theorem} The following are equivalent for any set $\Sigma\cup\{x\ank{k} y\}$ of simple $k$-anonymity atoms (for varying $k$):
\begin{enumerate}
\item $\Sigma\vdash x\ank{k} y$. 
\item $M\models_X\Sigma$ implies $M\models_X x\ank{k} y$ for every model $M$ and team $X$.
\end{enumerate}
\end{theorem}

\begin{proof} Suppose first $\Sigma\vdash x\ank{k} y$ and $M\models_X\Sigma$. All the rules (A1)-(A4) preserve truth. Hence $M\models_X x\ank{k} y$.
Suppose then $\Sigma\nvdash x\ank{k} y$, where $k\ge 2$. By the validity of (A3) we may assume $y\notin x$.  Let $x=\langle x_1,\ldots,x_n\rangle$. 
Let $W$ be the set of all variables in $\Sigma\cup\{x\ank{k} y\}$. Let $M$ be a model the domain of which contains the natural numbers, and let
$$
X=\{s\in \oN^W : \bigvee_{1\le i\le n}(s(x_i)\ne 0)\vee
(s(y)\in\{0,\ldots, k-2\})\}.
$$
This team, as well as $M$, are infinite, but the set $\Sigma$ may contain $u\ank{k} z$ for all $k\in\oN$, so the infinity of $X$ is unavoidable.
As in the proof of Theorem~\ref{main} we can show that
 $M\not\models_X x\ank{k} y$, namely,  let $s(w)=0$ for all $w\in W$. If $M\models_X x\ank{k} y$, then there are $k$ assignments  $s'\in X$ such that $s'(x_i)=s(x_i)=0$ for $1\le i\le n$ and $s'(y)$ is different for each $s'$. But this contradicts the definition of $X$.
%
Likewise, as in the proof of Theorem~\ref{main}, we can show that $M\models_X \Sigma$, namely, let ${u\ank{l} v}\in \Sigma$. We may assume, w.l.o.g., $v\notin u$. Let $s\in X$. By considering different cases according to whether
${x\subseteq u}$ or ${x\not\subseteq u}$, we can in both cases construct $l$ different $s'\in X$ such that $s'(u)=s(u)$ and $s'(v)$ is different for each $s'$.
\end{proof}

We have discussed only inferences among atomic anonymity statements. What about more complex formulas? Unfortunately consequence in  full first order logic with anonymity atoms is badly non-axiomatizable, a consequence of results in \cite{galthesis}. We can still investigate fragments of first order logic with anonymity atoms and look for fragments of low complexity. There are some early results in this direction in \cite{https://doi.org/10.48550/arxiv.2204.00576}.


\end{document}